\documentclass[11pt]{article}

\usepackage[margin=1in]{geometry}
\usepackage{amsfonts,amssymb,amsmath,amsthm,mathtools}
\usepackage[usenames,dvipsnames]{xcolor}
\usepackage{booktabs}
\usepackage{dsfont}
\usepackage{tikz}

\usepackage[pagebackref]{hyperref}
\hypersetup{
    pdftitle={}, 
    pdfauthor={}, 
    colorlinks=true, 
    linkcolor=blue, 
    citecolor=blue, 
    filecolor=blue, 
    urlcolor=blue 
}

\allowdisplaybreaks

\renewcommand{\backref}[1]{}

\renewcommand{\backrefalt}[4]{%
\small
\ifcase #1 %
\or
[p.\ #2]%
\else
[pp.\ #2]%
\fi}

\makeatletter
\newcommand{\para}{%
  \@startsection{paragraph}{4}%
  {\z@}{1.5ex \@plus .5ex \@minus .3ex}{-1em}%
  {\normalfont\normalsize\bfseries}%
}
\makeatother

\newtheorem{theorem}{Theorem}
\newtheorem{lemma}[theorem]{Lemma}
\newtheorem{corollary}[theorem]{Corollary}

\newtheorem{fact}[theorem]{Fact}

\newtheorem{conjecture}[theorem]{Conjecture}

\theoremstyle{definition}
\newtheorem{definition}[theorem]{Definition}

\newtheoremstyle{part}
  {-0.2\topsep}   
  {\topsep}   
  {\itshape}  
  {0pt}       
  {\bfseries} 
  {.}         
  {5pt plus 1pt minus 1pt} 
  {}          

\theoremstyle{part}
\newtheorem{factpart}{Fact}[theorem]

\newcommand{\eq}[1]{\hyperref[eq:#1]{(\ref*{eq:#1})}}
\renewcommand{\sec}[1]{\hyperref[sec:#1]{Section~\ref*{sec:#1}}}
\newcommand{\thm}[1]{\hyperref[thm:#1]{Theorem~\ref*{thm:#1}}}
\newcommand{\lem}[1]{\hyperref[lem:#1]{Lemma~\ref*{lem:#1}}}
\newcommand{\prop}[1]{\hyperref[prop:#1]{Proposition~\ref*{prop:#1}}}
\newcommand{\cor}[1]{\hyperref[cor:#1]{Corollary~\ref*{cor:#1}}}
\newcommand{\fig}[1]{\hyperref[fig:#1]{Figure~\ref*{fig:#1}}}
\newcommand{\tab}[1]{\hyperref[tab:#1]{Table~\ref*{tab:#1}}}
\newcommand{\alg}[1]{\hyperref[alg:#1]{Algorithm~\ref*{alg:#1}}}
\newcommand{\app}[1]{\hyperref[app:#1]{Appendix~\ref*{app:#1}}}
\newcommand{\defn}[1]{\hyperref[def:#1]{Definition~\ref*{def:#1}}}
\newcommand{\clm}[1]{\hyperref[clm:#1]{Claim~\ref*{clm:#1}}}
\newcommand{\fct}[1]{\hyperref[fact:#1]{Fact~\ref*{fact:#1}}}

\newcommand*{\fullref}[1]{\hyperref[{#1}]{\autoref*{#1}:~\nameref*{#1}}}

\newcommand*{\fullbref}[1]{\hyperref[{#1}]{\autoref*{#1} (\nameref*{#1})}}

\newcommand{\B}{\{0,1\}}

\newcommand{\SINK}{\textsf{Sink}}
\newcommand{\XOR}{\textsf{Xor}}

\newcommand{\Eq}[1]{\textsc{EQ}_{#1}}

\newcommand{\X}{\mathcal{X}}
\newcommand{\Y}{\mathcal{Y}}

\newcommand{\U}{\mathcal{U}}

\newcommand{\srd}{R}

\renewcommand{\(}{\left(}
\renewcommand{\)}{\right)}
\newcommand{\<}{\langle}
\renewcommand{\>}{\rangle}
\newcommand{\id}{\mathbb{I}}

\DeclareMathOperator{\rank}{rk}

\DeclareMathOperator{\dom}{{dom}}

\newcommand{\Q}{\mathrm{Q}^*}

\DeclareMathOperator{\IC}{IC}
\DeclareMathOperator{\CC}{CC}

\newcommand{\defeq}{\coloneqq} 

\newcommand{\eps}{\varepsilon}
\renewcommand{\epsilon}{\varepsilon}
\renewcommand{\Pr}{\mathrm{Pr}}
\newcommand{\E}{\mathcal{E}}
\newcommand{\F}{\mathrm{F}}
\newcommand{\I}{\mathrm{I}}

\newcommand{\err}{\mathrm{err}}

\newcommand {\fn} [2] {\ensuremath{ #1 \minusspace \br{ #2 } }}
\newcommand {\Fn} [2] {\ensuremath{ #1 \minusspace \Br{ #2 } }}
\newcommand {\minusspace} {\: \! \!}

\newcommand {\br} [1] {\ensuremath{ \left( #1 \right) }}
\newcommand {\Br} [1] {\ensuremath{ \left[ #1 \right] }}
\newcommand {\bra} [1] {\ensuremath{ \left\langle #1 \right| }}
\newcommand {\ket} [1] {\ensuremath{ \left| #1 \right\rangle }}
\newcommand {\ketbratwo} [2] {\ensuremath{ \left| #1 \middle\rangle \middle\langle #2 \right| }}
\newcommand {\ketbra} [1] {\ketbratwo{#1}{#1}}
\newcommand {\Tr} {\ensuremath{ \mathrm{Tr} }}
\newcommand {\norm} [1] {\ensuremath{ \left\| #1 \right\| }}
\newcommand {\normsub} [2] {\ensuremath{ \norm{#1}_{#2} }}
\newcommand {\onenorm} [1] {\normsub{#1}{1}}
\newcommand {\cspace} [1] {\ensuremath{\mathnormal{#1}}}
\newcommand {\set} [1] {\ensuremath{ \left\lbrace #1 \right\rbrace }}

\newcommand {\ent} [1] {\fn{\mathrm{H}}{#1}}

\newcommand {\suppress}[1]{}

\newcommand{\expec}{\mathbb{E}}

\def\D{\mathcal{D}}
\def\L{\mathcal{L}}

\def\H{\mathcal{H}}
\newcommand{\BR}{\mathrm{B}}

\DeclareMathOperator{\QIC}{QIC}
\DeclareMathOperator{\HQIC}{HQIC}
\DeclareMathOperator{\SQIC}{SQIC}

\newcommand{\QCC}{\mathrm{QCC}}
\newcommand{\Qent}{\mathrm{Q}^*}

\newcommand{\beq}{\begin{equation}}
\newcommand{\eeq}{\end{equation}}
\newcommand{\beqst}{\begin{equation*}}
\newcommand{\eeqst}{\end{equation*}}
\newcommand{\beqar}{\begin{eqnarray}}
\newcommand{\eeqar}{\end{eqnarray}}
\newcommand{\beqarst}{\begin{eqnarray*}}
\newcommand{\eeqarst}{\end{eqnarray*}}

\usepackage{framed}

\newcommand{\kb}[2]{| #1\rangle\!\langle #2 |}

\newcommand {\prob} [1] {\Fn{\Pr}{#1}}

\newcommand{\cO}{\mathcal{O}}

\begin{document}

\title{Quantum Log-Approximate-Rank Conjecture is also False}

\author{
Anurag Anshu \footnote{Institute for Quantum Computing and Department of Combinatorics and Optimization, University of Waterloo, and Perimeter Institute for Theoretical Physics, \texttt{aanshu@uwaterloo.ca}} \qquad
Naresh Goud Boddu\footnote{Center for Quantum Technologies, National University of Singapore, \texttt{e0169905@u.nus.edu}}\qquad
Dave Touchette\footnote{Institute for Quantum Computing and Department of Combinatorics and Optimization, University of Waterloo, and Perimeter Institute for Theoretical Physics, touchette.dave@gmail.com}
}

\hypersetup{pageanchor=false} 
\date{}
\maketitle

\begin{abstract}
In a recent breakthrough result, Chattopadhyay, Mande and Sherif [ECCC TR18-17] showed an exponential separation between the log approximate rank and randomized communication complexity of a total function $f$, hence refuting the log approximate rank conjecture of Lee and Shraibman [2009]. We provide an alternate proof of their randomized communication complexity lower bound using the information complexity approach. Using the intuition developed there, we derive a polynomially-related quantum communication complexity lower bound using the quantum information complexity approach, thus providing an exponential separation between the log approximate rank and quantum communication complexity of $f$. Previously, the best known separation between these two measures was (almost) quadratic, due to Anshu, Ben-David, Garg, Jain, Kothari and Lee [CCC, 2017]. This settles one of the main question left open by Chattopadhyay,  Mande and  Sherif, and refutes the quantum log approximate rank conjecture of Lee and Shraibman [2009]. Along the way, we develop a Shearer-type protocol embedding for product input distributions that might be of independent interest.
\end{abstract}


\thispagestyle{empty}
\clearpage
\section{Introduction}
\label{sec:intro}
\hypersetup{pageanchor=true} 
\setcounter{page}{1}

Communication complexity concerns itself with characterizing the minimum number of bits that distributed parties need to exchange in order to accomplish a given task (such as computing a function $F$). Over the years, it has established striking connections with various areas of complexity theory and information theory, providing tools for solving central problems in such domains. Since it is in general hard to pin down precisely the communication cost of a task, various lower bound methods have been developed over the years. One such method is the logarithm of the rank of the matrix $M_F$ that encodes the values the function $F$ takes on various inputs. More precisely, this matrix is defined as $M_F(x,y) = F(x,y)$. The following well known conjecture posits that this lower bound is polynomially tight for the deterministic communication complexity of $F$.  

\begin{conjecture}[Log-Rank Conjecture,~\cite{LS88}]
\label{logrankconj}
There exists a universal constant $\alpha$ such that the deterministic communication complexity of every total Boolean function $F$ is $ \cO(\log^{\alpha}(\rank(M_F)))$.  
\end{conjecture}

See Ref.~\cite{ANS18} and reference therein for more details about this and the other conjectures discussed in this work.
A natural randomized analogue of Conjecture \ref{logrankconj} is the following, comparing randomized communication complexity to the logarithm of the approximate rank rather than actual rank of $M_F$. (See Section~\ref{sec:ccomm} for definitions.)

\begin{conjecture}
[Log-Approximate-Rank Conjecture,~\cite{LS09}]
\label{logaprankconj}
 There exists a universal constant $\alpha$ such
that the randomized communication complexity (with error $\frac{1}{3}$) of every total Boolean function $F$ is $ O(\log^{\alpha}(\rank_{1/3}(M_F)))$. 
\end{conjecture}

In a recent breakthrough work~\cite{ANS18}, Chattopadhyay, Mande and Sherif establish that Conjecture \ref{logaprankconj} is false by exhibiting a function with an exponential separation between the randomized communication complexity (with constant error) and Log-Approximate-Rank. Their function is a composition of the $2$-bit $\XOR$ function and a function that they call $\SINK$.  The work~\cite{ANS18} asked if their function had implications for the following quantum version of Conjecture \ref{logaprankconj}.

\begin{conjecture}[Quantum Log-Approximate-Rank Conjecture,~\cite{LS09}]
\label{Qlogaprankconj} 
There exists a universal constant $\alpha$, such that the quantum communication complexity of every total Boolean function $F$ is $ O(\log^{\alpha}(\rank_{1/3}(M_F)))$.
\end{conjecture}

Here we prove that Conjecture \ref{Qlogaprankconj} is false as well.  Before proceeding to the statement of our main result, we define the $\SINK$ function.

\begin{definition}[$\SINK$~\cite{ANS18}] 
\label{sink}
$\SINK$ function is defined on a complete directed graph of $m$ vertices, using $m \choose 2$
variables $z_{i,j}, i < j \in [m]$, in the following way. Let $z_{i,j }= 1$ if there is a directed edge from vertex $v_i$ to $v_j$ and $z_{i,j }= 0$ if there is a  directed edge from vertex $v_j$ to $v_i$. The function $\SINK$ computes whether or not there is a
sink in the graph. In other words, $\SINK(z) = 1 $ iff $ \exists i \in  [m]$ such that all edges adjacent to $v_i$ are incoming.
\end{definition}

The function of interest for communication complexity is $\SINK \circ \XOR^{\otimes {m\choose 2}}$, where each $\XOR$ takes as input one bit from Alice and one from Bob. For simplicity of notation, we will denote this function as $\SINK\circ \XOR$. Our main theorem is as follows, which lower bounds the quantum information complexity ($\QIC$) of $\SINK \circ \XOR$.

\begin{theorem}
\label{theo:main}
Any t-round entanglement assisted protocol for $\SINK \circ \XOR$ achieving error $1/5$ satisfies $\QIC(\Pi, \mu^{\otimes {m \choose 2}}) \in \Omega (\frac{m}{t^2})$, with $\mu$ being the uniform distribution on 1+1 bits \footnote{A random variable on $a+b$ bits takes values over $a$ bits on Alice's side and $b$ bits on Bob's side.}.
\end{theorem}

The desired lower bound on entanglement assisted quantum communication complexity ($\Qent_{\frac{1}{3}}$) of $\SINK \circ \XOR$ follows by optimizing $\max (t, m/t^2)$ over the number of round $t$. 

\begin{corollary}
\label{cor:qlb}
It holds that $\Qent_{1/3}(\SINK \circ \XOR) \in \Omega (m^{1/3})$.
\end{corollary}

Hence, combining with the following upper bound on the log-approximate-rank due to Ref.~\cite{ANS18}, the $\SINK \circ \XOR$ function witnesses an exponential separation between log-approximate-rank and quantum communication, and refutes the quantum log-approximate-rank conjecture of Lee and Shraibman \cite{LS09}.
\begin{theorem}[~\cite{ANS18}]
\label{theorem1}
It holds that 
\begin{enumerate}
\item $\log\rank_{1/3}(M_{\SINK\circ\XOR}) \leq 4\log m + o(\log m)$
\item  $\log\rank^{+}_{1/3}(M_{\SINK\circ\XOR}) = O(\log^2 m)$.
\end{enumerate}
\end{theorem}
In a subsequent version of~\cite{ANS18}, Chattopadhyay et. al. improved the upper bound on $\log\rank^{+}_{1/3}(M_{\SINK\circ\XOR})$ to $O(\log m)$.

\subsection{Independent work}

Sinha and de Wolf~\cite{SW18} used the fooling distribution method, in independent and simultaneous work, to obtain the same $\Omega(m^{1/3})$ lower bound on the quantum communication complexity of $\SINK\circ\XOR$. This differs from our techniques which we describe below.

\subsection{Proof overview}

At a high-level, our argument follows the well-established \textit{information complexity} approach~\cite{KNTZ07, CSWY01, BJKS04, JRS03b,  BarakBCR10}. We view a given function $f$ as some composition of many instances of a simpler component function $g$, and argue through a direct sum property a reduction from $g$ to $f$. This is achieved by embedding inputs to $g$ into inputs to $f$, where the remaining inputs to $f$ are sampled from some suitable distribution in order to achieve the desired direct sum property. Following this, we show a lower bound on the information complexity for g. 

In the present context, $\SINK\circ\XOR$ is a composition of many instances of the Equality function, in a way that the input bits are shared across the instances. In Ref.~\cite{ANS18}, the authors use Shearer's lemma to handle such overlap between the inputs across the instances and derive a corruption lower bound. For the reduction from $\SINK\circ\XOR$ to Equality,  we also wish to use a Shearer-type inequality. We further argue that a lower bound on information complexity of Equality (for protocols that make small error in the worst case) under uniform distribution implies a lower bound on information complexity of $\SINK\circ\XOR$. But it is not clear, a priori, that Equality should have high information cost under that distribution, as this function has trivial communication complexity under the uniform distribution. It turns out that the cut-and-paste argument of Anshu, Belovs, Ben-David, G{\"o}{\"o}s, Jain, Kothari,  Lee and  Santha~\cite{ABB+16a} yields a constant lower bound on information complexity of good protocols for Equality, even under the uniform distribution. 

Broadly, our quantum lower bound proceeds along lines similar to above. The quantum cut-and-paste argument of Anshu,  Ben-David, Garg,  Jain,  Kothari and Lee~\cite{ABGJKL17} in the quantum setting yields a round dependent lower bound on the \textit{quantum information complexity} (QIC)~\cite{KNTZ07, JRS03, JN14, Tou15, KLLR16} of good protocols for Equality, even under the uniform distribution. But the quantum version of the embedding argument requires new methods. In the classical setting, using classical information cost $\IC$, as soon as we have Alice and Bob privately sample the remaining inputs,  the Shearer-type embedding follows almost directly from a Shearer like inequality for information~\cite{GanorKR:2015}. In the quantum setting, we would similarly like to use a Shearer-type inequality for quantum information~\cite{ATYY16}. However,  it is not immediately clear how to make the protocol embedding work for quantum information cost $\QIC$. We instead settle on an alternate notion of quantum information cost  (variants of which have appeared before~\cite{JRS05, JN14, LT17, ATYY16}) that works well only for product input distributions. The argument then goes through by carefully using this notion, and it is equivalent to $\QIC$ up to a round-dependent factor. 
What we get is a Shearer-type embedding protocol for product input distributions that allows some specific pre-processing of the inputs. We provide such a general version in Section~\ref{sec:qu-shearer} in the quantum setting, while we give a more direct proof in the classical setting.

Hence, overall we get a round dependent lower bound on the quantum information complexity of $\SINK\circ\XOR$, and the round independent lower bound on quantum communication complexity follows by optimizing over the number of rounds in any good protocol.

\section{Preliminaries and notation}
\label{sec:prelim}

For integer $n \geq 1$, let $[n]$ represent the set $\{1, 2, ..., n\}.$
Let $\mathcal{X} $ and $ \mathcal{Y}$ be finite sets and $k$ be a natural number. Let
$\mathcal{X}^k $ be the set $\mathcal{X} \times ... \times \mathcal{X}$, the cross product of $\mathcal{X}$, $k$ times. Let $\mu$ be a probability distribution on $\mathcal{X}$. Let $\mu(x)$ represent
the probability of $x \in \mathcal{X}$ according to $\mu$. We write $X \sim \mu$ to denote that the random variable $X$ is distributed according to $\mu$. We use the same symbol
to represent a random variable and its distribution whenever it
is clear from the context. 
The expectation value of function $f$
on $X$ is defined as $\mathbb{E}_{x 	\leftarrow X}[f(x)] = \sum_{x \in \mathcal{X}  }\Pr(X=x)f(x)$
where $x 	\leftarrow X$ means that $x$ is drawn according to the distribution
of $X$. We say $X$ and $Y$ are independent iff for each $x \in \mathcal{X}, y \in \mathcal{Y}:\Pr(XY=xy) = \Pr(X=x) \cdot \Pr(Y=y)$. 
For joint random variables $XY$, $Y^x$ will denote the distribution of $Y|X=x$.

We now introduce some quantum information theoretic notation. We assume the reader is familiar with standard concepts in quantum computing~\cite{NC00,Wil12, Wat16}.

Let $\H$ be a finite-dimensional complex Euclidean space, i.e.,  $\mathbb{C}^n$ for some positive integer $n$ with the usual complex inner product $\langle \cdot, \cdot \rangle$, which is defined as $\langle u,v \rangle = \sum_{i=1}^n u_i^* v_i$. We will also refer to $\H$ as an Hilbert space. We will usually denote vectors in $\H$ using bra-ket notation, e.g., $|\psi\> \in \H$.

The $\ell_1$ norm (also called the trace norm) of an operator $X$ on $\H$ is $\onenorm{X}\defeq \Tr (\sqrt{X^{\dag}X})$, which is also equal to (vector) $\ell_1$ norm of the vector of singular values of $X$. 
A {\em quantum state} (or a {\em density matrix} or simply a {\em state}) $\rho$ is a positive semidefinite matrix on $\H$ with $\Tr(\rho)=1$. The state $\rho$ is said to be a {\em pure state} if its rank is $1$, or equivalently if $\Tr(\rho^2)=1$, and otherwise it is called a {\em mixed state}. 
 Let $\ket{\psi}$ be a unit vector on $\H$, that is $\langle \psi|\psi \rangle=1$.  With some abuse of notation, we use $\psi$ to represent the vector $|\psi\>$ and also the density matrix $\ketbra{\psi}$, associated with $\ket{\psi}$. Given a quantum state $\rho$ on $\H$, the {\em support of $\rho$}, denoted $\text{supp}(\rho)$, is the subspace of $\H$ spanned by all eigenvectors of $\rho$ with nonzero eigenvalues.

A {\em quantum register} $A$ is associated with some Hilbert space $\H_A$. Define $|A| \defeq \log \dim(\H_A)$. Let $\L(A)$ represent the set of all linear operators on $\H_A$. We denote by $\D(A)$ the set of density matrices on the Hilbert space $\H_A$. We use subscripts (or superscripts according to whichever is convenient) to denote the space to which a state belongs, e.g, $\rho$ with subscript $A$ indicates $\rho_A \in \H_A$. If two registers $A$ and $B$ are associated with the same Hilbert space, we represent this relation by $A\equiv B$.  For two registers $A$ and $B$, we denote the combined register as $AB$, which is associated with Hilbert space $\H_A \otimes \H_B$.  For two quantum states $\rho\in \D(A)$ and $\sigma\in \D(B)$, $\rho\otimes\sigma \in \D(AB)$ represents the tensor product (or Kronecker product) of $\rho$ and $\sigma$. The identity operator on $\H_A$ is denoted $\id_A$. 

Let $\rho_{AB} \in \D(AB)$. We define the {\em partial trace with respect to $A$} of $\rho_{AB}$ as
\[ \rho_{B} \defeq \Tr_{A}(\rho_{AB})
\defeq \sum_{i} (\bra{i} \otimes \id_{\cspace{B}})
\rho_{AB} (\ket{i} \otimes \id_{\cspace{B}}) , \]
where $\set{\ket{i}}_i$ is an orthonormal basis for the Hilbert space $\H_A$.
The state $\rho_B\in \D(B)$ is referred to as a {\em reduced density matrix} or a {\em marginal state}. Unless otherwise stated, a missing register from subscript in a state will represent partial trace over that register. Given  $\rho_A\in\D(A)$, a {\em purification} of $\rho_A$ is a pure state $\rho_{AB}\in \D(AB)$ such that $\Tr_{B}(\rho_{AB})=\rho_A$. Any quantum state has a purification using a register $B$ with $|B|\leq |A|$. The purification of a state, even for a fixed $B$, is not unique as any unitary applied on register $B$ alone does not change $\rho_A$.

An important class of states that we will consider are the {\em classical quantum states}. They are of the form $\rho_{AB} = \sum_a \mu(a) \ketbra{a}_A\otimes \rho^a_B$, where $\mu$ is a probability distribution. In this case, $\rho_A$ can be viewed as a probability distribution and we shall continue to use the notations that we have introduced for probability distribution, for example,  $\expec_{a\leftarrow A}$ to denote the average $\sum_a \mu(a)$. 

A quantum {\em super-operator} (or a {\em quantum channel} or a {\em quantum operation}) $\E: A\rightarrow B$ is a completely positive and trace preserving (CPTP) linear map (mapping states from $\mathcal{D}(A)$ to states in $\mathcal{D}(B)$). The identity operator in Hilbert space $\H_A$ (and associated register $A$) is denoted $\id_A$.  A {\em unitary} operator $\U_A:\H_A \rightarrow \H_A$ is such that $\U_A^{\dagger}\U_A = \U_A \U_A^{\dagger} = \id_A$. The set of all unitary operations on register $A$ is  denoted by $\mathcal{U}(A)$. 

A $2$-outcome quantum measurement is defined by a collection $\{M, \id - M\}$, where $0 \preceq M \preceq \id$ is a positive semidefinite operator, where $A\preceq B$ means $B-A$ is positive semidefinite. Given a quantum state $\rho$, the probability of getting outcome corresponding to $M$ is $\Tr(\rho M)$ and getting outcome corresponding to $\id - M$ is $1-\Tr(\rho M)$.

\subsubsection{Distance measures for quantum states}

We now define the distance measures we use and some properties of these measures. Before defining the distance measures, we introduce the concept of {\em fidelity} between two states, which is not a distance measure but a similarity measure. Note that all the notions introduced below also apply to classical random variables, when viewed as diagonal quantum states in some basis.

\begin{definition}[Fidelity]
 Let $\rho_A,\sigma_A \in \D(A)$ be quantum states. The fidelity between $\rho$ and $\sigma$ is defined as
$$\F(\rho_A,\sigma_A)\defeq\onenorm{\sqrt{\rho_A}\sqrt{\sigma_A}}.$$
\end{definition}

For two pure states $|\psi\>$ and $|\phi\>$, we have $\F(|\psi\>\<\psi|,|\phi\>\<\phi|) = |\<\psi|\phi\>|$. We now introduce the two distance measures we use.

\begin{definition}[Distance measures]
 Let $\rho_A,\sigma_A \in \D(A)$ be quantum states. We define the following distance measures between these states.
\begin{align*}
\text{Trace distance:}& \quad \Delta(\rho_A,\sigma_A) \defeq \frac{1}{2}\|\rho_A-\sigma_A\|_1  \\
\text{Bures metric:}& \quad \BR(\rho_A,\sigma_A) \defeq \sqrt{1-\F(\rho_A,\sigma_A)}. 
\end{align*}
\end{definition}

Note that for any two quantum states $\rho_A$ and $\sigma_A$, these distance measures lie in $[0,1]$. The distance measures are $0$ if and only if the states are equal, and the distance measures are $1$ if and only if the states have orthogonal support, i.e., if $\rho_A \sigma_A = 0$.

Conveniently, these measures are closely related.

\begin{fact}\label{fact:deltabures}For all quantum states $\rho_A,\sigma_A \in \D(A)$, we have
\begin{align*}
&\quad \BR^2(\rho_A,\sigma_A) \leq \Delta(\rho_A,\sigma_A) \leq \sqrt{2} \cdot \BR(\rho_A,\sigma_A). 
\end{align*}
\end{fact}

\begin{proof} The Fuchs-van de Graaf inequalities~\cite{FvdG99,Wat16} state that 
\begin{align*}
&\quad 1-\F(\rho_A,\sigma_A) \leq \Delta(\rho_A,\sigma_A) \leq \sqrt{1-\F^2(\rho_A,\sigma_A)}. 
\end{align*}
Our fact follows from this and the relation $1-\F^2(\rho_A,\sigma_A) \leq 2-2\F(\rho_A,\sigma_A)$.
\end{proof}


We now review some properties of the Bures metric.

\begin{fact}[Facts about Bures metric]\label{fact:relation-inf} 
\end{fact}
\begin{factpart}[Triangle inequality \cite{Bures69}]\label{fact:triangle} The following triangle inequality and a weak triangle inequality hold for the Bures metric and the square of the Bures metric. 
\begin{enumerate}
\item $\BR(\rho_A,\sigma_A) \leq \BR(\rho_A,\tau_A) + \BR(\tau_A,\sigma_A).$
\item $\BR^2(\rho_A^1, \rho_A^{t+1}) \le t \cdot \sum_{i=1}^t \BR^2(\rho_A^i, \rho_A^{i+1}).$
\end{enumerate}
\end{factpart}

\begin{factpart}[Averaging over classical registers]\label{fact:subadd} For classical-quantum states $\theta_{XB},\theta'_{XB}$ with $\theta_{X} = \theta'_{X}$, we have
$$\BR^2(\theta_{XB},\theta'_{XB}) = \expec_{x\leftarrow X} [\BR^2(\theta^x_B, \theta'^x_B)].$$
\end{factpart}

Finally, an important property of both these distance measures is monotonicity under quantum operations \cite{lindblad75,barnum96}.

\begin{fact}[Monotonicity under quantum operations]
\label{fact:monotonicitydistance}
For quantum states $\rho_A$, $\sigma_A \in \D(A)$, and a quantum operation $\E(\cdot):\L(A)\rightarrow \L(B)$, it holds that
\begin{align*}
	\Delta(\E(\rho_A) , \E(\sigma_A)) \leq \Delta(\rho_A,\sigma_A) \quad \mbox{and} \quad \BR(\E(\rho_A),\E(\sigma_A)) \leq \BR(\rho_A,\sigma_A),
\end{align*}
with equality if $\E$ is unitary. 
In particular, for bipartite states $\rho_{AB},\sigma_{AB}\in \D(AB)$, it holds that
\begin{align*}
	\Delta(\rho_{AB},\sigma_{AB}) \geq \Delta(\rho_A,\sigma_A) \quad \mbox{and} \quad \BR(\rho_{AB},\sigma_{AB}) \geq \BR(\rho_A,\sigma_A).
\end{align*}
\end{fact}

\subsubsection{Mutual information}

We start with the following fundamental information theoretic quantities. We refer the reader to the excellent sources for quantum information theory \cite{Wil12, Wat16} for further study.

\begin{definition}\label{def:relentropy}
Let $\rho_A \in \D(A)$ be a quantum state. 
We then define the following.
\begin{align*}
\text{von Neumann entropy:}& \quad \ent{\rho_A} \defeq - \Tr(\rho_A\log\rho_A) .
\end{align*}
\end{definition}

We now define mutual information and conditional mutual information.

\begin{definition}[Mutual information]
\label{def:entropy}
Let  $\rho_{ABC}\in\D(ABC)$ be a quantum state. We define the following measures.
\begin{align*}
\text{Mutual information:}& \quad \I(A:B)_{\rho}\defeq \ent{\rho_A} + \ent{\rho_B}-\ent{\rho_{AB}} .\\
\text{Conditional mutual information:}& \quad \I(A:B~|~C)_{\rho}\defeq \I(A:BC)_{\rho}-\I(A:C)_{\rho}.
\end{align*}
\end{definition}

We will need the following basic properties.

\begin{fact}[Properties of $\I$]
Let $\rho_{ABC}\in\D(ABC)$ be a quantum state. We have the following.
\end{fact}
\begin{factpart}[Nonnegativity]\label{fact:nonneg}
\begin{align*}
\I(A:B)_{\rho}  \geq 0 &\text{ and }\I(A:B~|~C)_{\rho}  \geq 0.  
\end{align*}
If $\rho_{AB} = \rho_A \otimes \rho_B$ is a product state, then
\begin{align*}
\I(A:B) = 0.
\end{align*}
\end{factpart}

\begin{factpart}[Chain rule]\label{fact:chain-rule}
$\I(A:BC)_{\rho} = \I(A:C)_{\rho} + \I(A:B~|~C)_{\rho}  = \I(A:B)_{\rho} + \I(A:C~|~B)_{\rho}.$
\end{factpart}
\begin{factpart}[Monotonicity]\label{fact:mono}  For a quantum operation $\E(\cdot):\L(A)\rightarrow \L(B)$,
$\I(A:\E(B)) \le \I(A:B)$ with equality when $\E$ is unitary. In particular $\I(A:BC)_{\rho} \ge  \I(A:B)_{\rho} .$
\end{factpart}

\begin{factpart}[Averaging over conditioning register]\label{fact:IvsB} For classical-quantum state (register $X$ is classical) $\rho_{XAB}$:
\begin{align*}
\I(A:B|X)_{\rho} &= \mathbb{E}_{x \leftarrow X}  \I(A:B)_{\rho^x} . 
\end{align*}
\end{factpart}

The following lemma, known as the Average Encoding Theorem~\cite{KNTZ07}, formalizes the intuition that if a classical and a quantum registers are weakly correlated, then they are nearly independent.

\begin{lemma}
\label{lem:avenc}
For any $\rho_{XA} = \sum_x p_X(x) \cdot \kb{x}{x}_X \otimes \rho^x_{A}$ with a classical system $X$ and
 states $\rho^x_A$,
\begin{align}
	\sum_x p_X (x) \cdot \BR^2  \!\left(\rho^x_A, \rho_A \right) &
\quad \leq \quad  \I(X \!:\! A )_{\rho}
\enspace.
\end{align}

\end{lemma}

The following Shearer-type inequality for quantum information was shown in Ref.~\cite{ATYY16}. Classical variants appeared in \cite{GanorKR:2015, RaoS:2015}.

\begin{lemma}\label{lem:shearer}
	Consider registers $U_1,U_2,\ldots U_m, V$ and define $U\defeq U_1U_2\ldots U_m$. Consider a quantum state $\Psi_{UV}$ such that $\Psi_{U_1 U_2 \ldots U_m} = \Psi_{U_1}\otimes \Psi_{U_2}\otimes\ldots \otimes \Psi_{U_m}$. Let $S=\set{i_1,\ldots,i_{|S|}}\subseteq [m]$ be a random set picked  independently of $\Psi_{UV}$ satisfying $\prob{i\in S}\leq\frac{1}{k}$ for all $i$ and $U_S\defeq U_{i_1}U_{i_2}\ldots U_{i_{|S|}}$. Then it holds that
	\[\I(U_S:V~|~ S)_{\Psi}\leq\frac{\I(U:V)_{\Psi}}{k},\]
\end{lemma}

\subsection{Classical communication complexity}
\label{sec:ccomm}

Let $f: \X \times \Y \to \{0,1\} $ be a total function (that is, its value is defined on every input) and $\epsilon \in (0, 1)$. In a two-party communication task, Alice is given an input $x \in \X$, Bob is given $y \in \Y$ and the task is to compute $f(x,y)$ by exchanging as few bits as possible. The parties are allowed to possess pre-shared randomness ($R$) and private randomness ($\srd_A$, $\srd_B$). Without loss of generality, we can assume that Alice communicates first and also gives the final output.  The communication cost of a protocol $\Pi$, denoted by $CC(\Pi)$, is the maximum number of bits the parties have to communicate over all possible inputs and values of the shared and private randomness. Let $R_{\eps}(f)$ represent the two-party randomized
communication complexity of $f$ with worst case error $\epsilon$, i.e., the communication of the best
two-party randomized protocol for $f$ with error at most $\epsilon$ over any input $(x, y)$. Worst-case error of the protocol $\Pi$ over the inputs is denoted by $\err(\Pi)$.

\begin{definition}[XOR function]

A function $F : \{0, 1\}^n  \times \{0, 1\}^n \to \{0, 1\}$ is called an XOR function if there exists a function $f : \{0, 1\}^n \to \{0,1\}$ such that $F(x_1, ... , x_n, y_1, ... , y_n) = f(x_1 \oplus y_1, ... , x_n \oplus y_n)$ for all
$x, y \in \{0, 1\}^n$. We denote $F = f \circ XOR$.

\end{definition}

\begin{definition}[Rank]

The rank of a matrix
$M$, denoted by $\rank(M)$  is the minimum integer $k$ for which there exist $k$ rank 1 matrices such that $M =
\sum^{k}_{i=1}M_i$.

\end{definition}

\begin{definition}[Non-negative Rank]
The non-negative rank of a matrix
$M$, denoted by $\rank^{+}(M)$  is the minimum integer $k$ for which there exist $k$ rank 1 matrices with non-negative entries such that $M = \sum^{k}_{i=1}M_i$.
\end{definition}

\begin{definition}[Approximate rank]
Let $\epsilon \in [0,1/2)$ and $M$ be an $|\X| \times |\Y|$ matrix.  The $\epsilon$-approximate rank of $M$ is 
defined as
\[
\rank_\epsilon(M) = \min_{\tilde M}\  \{\rank(\tilde M) : \forall x\in\X, y\in \Y, \, |\tilde M(x,y) - M(x,y)| \le \epsilon\}.
\] 
\end{definition}

\begin{definition}[Approximate  non-negative rank]
Let $\epsilon \in [0,1/2)$ and $M$ be an $|\X| \times |\Y|$ matrix.  The $\epsilon$-approximate non-negative rank of $M$ is 
defined as
\[
\rank^{+}_\epsilon(M) = \min_{\tilde M}\  \{\rank^{+}(\tilde M) : \forall x\in\X, y\in \Y, \, |\tilde M(x,y) - M(x,y)| \le \epsilon\}.
\] 
\end{definition}

\begin{definition}[Distributional Information Complexity]
Distributional information complexity of a randomized protocol $\Pi$ with respect to a distribution $XY \sim \mu$ is defined as 

$$\IC(\Pi, \mu) = \I(X: \Pi|YRR_B) + \I(Y:\Pi|XRR_A).$$
\end{definition}

\begin{definition}[Max Distributional Information Complexity]

Max-distributional information complexity of a randomized protocol $\Pi$ is defined as 

$$\IC(\Pi) = \max_{\mu}\IC(\Pi, \mu).$$
 
\end{definition}

\begin{definition}[Information Complexity of a function]

Information complexity of a  function $f$ is defined as 

$$\IC(f) = \inf_{\Pi : \err(\Pi) \leq \eps}\IC(\Pi).$$
 
\end{definition}

Note that since one bit of communication can hold at most one bit of information, for any
protocol $\Pi$ and distribution $\mu$ we have $\IC(\Pi, \mu)  \leq \CC(\Pi)$. This implies that information complexity of a function is a lower bound on the randomized communication complexity of a function.

\begin{lemma}[Cut-and-paste lemma (Lemma 6.3 in ~\cite{BJKS04})]
\label{lemma:classicalcutandpaste}
Let $(x, y)$ and $(x', y')$ be two inputs to a randomized protocol $\Pi$. Then 
$$ \BR( \Pi(x,y) , \Pi(x',y')   )= \BR( \Pi(x,y') , \Pi(x',y)   ).$$

\end{lemma}

\begin{fact}[Pythagorean property  (Lemma 6.4 in ~\cite{BJKS04})]
\label{lemma:classicalpythagorean}

 Let $(x, y)$ and $(x', y')$ be two inputs to a randomized protocol $\Pi$. Then 

$$ \BR^{2}( \Pi(x,y') , \Pi(x',y')   ) + \BR^{2}( \Pi(x,y) , \Pi(x',y)   )\leq 2 \BR^{2}( \Pi(x',y') , \Pi(x,y)   ).$$

\end{fact}

\subsection{Quantum communication complexity}
\label{sec:comm}

In quantum communication complexity, two players wish to compute a classical function $F\colon\X \times \Y \to \{0,1\}$ for some finite sets $\X$ and $\Y$. The inputs $x\in \X$ and $y\in\Y$ are given to two players Alice and Bob, and the goal is to minimize the quantum communication between them required to compute the function. 

While the players have classical inputs, the players are allowed to exchange quantum messages. Depending on whether or not we allow the players arbitrary shared entanglement, we get $\mathrm{Q}(F)$, bounded-error quantum communication complexity without shared entanglement and $\Q(F)$, for the same measure with shared entanglement. Obviously $\Q(F) \leq \mathrm{Q}(F)$.
In this paper we will only work with $\Q(F)$, which makes our results stronger since we prove lower bounds in this work. 

 An entanglement assisted quantum communication protocol $\Pi$ for a function is as follows. Alice and Bob start with preshared entanglement $\ket{\Theta_0}_{A_0B_0}$. Upon receiving inputs $(x,y)$, where Alice gets $x$ and Bob gets $y$, they exchange quantum messages. At the end of the protocol, Alice applies a two outcome measurement on her qubits and correspondingly outputs $1$ or $0$. Let $O(x,y)$ be the random variable corresponding to the output produced by Alice in $\Pi$, given input $(x,y)$. 

Let $\mu$ be a distribution over $\dom(F)$. 
Let inputs to Alice and Bob be given in registers $X$ and $Y$ in the state 
\begin{eqnarray}
\label{eq:classicalst}
\rho_\mu := \sum_{x,y}\mu(x,y)\ketbra{x}_X\otimes\ketbra{y}_Y.
\end{eqnarray}
Let these registers be purified by $R_X$ and $R_Y$ respectively, which are not accessible to either players. Denote
\begin{eqnarray}
\ket{\mu}_{X R_X Y R_Y} := \sum_{x,y}\sqrt{\mu(x,y)}\ket{xxyy}_{XR_X Y R_Y}.
\end{eqnarray}
Let Alice and Bob initially hold register $A_0,B_0$ with shared entanglement $\Theta_{0,A_0B_0}$. Then the initial state is 
\begin{eqnarray}
\label{eq:initialstate}
\ket{\Psi_0}_{XYR_XR_YA_0B_0} \defeq \ket{\mu}_{XYR_XR_Y}\ket{\Theta_0}_{A_0B_0}.
\end{eqnarray}

Alice applies a unitary $U^1: XA_0\rightarrow XA_1C_1$ such that the unitary acts on $A_0$ conditioned on $X$.  She sends $C_1$ to Bob. Let $B_1\equiv B_0$ be a relabeling of Bob's register $B_0$. He applies $U^2: YC_1B_1\rightarrow YC_2B_2$ such that the unitary acts on $C_1B_0$ conditioned on $Y$. He sends $C_2$ to Alice. Players proceed in this fashion for $t$ messages, for $t$ even, until the end of the protocol. At any round $r$, let the registers be $A_rC_rB_r$, where $C_r$ is the message register, $A_r$ is Alice's register and $B_r$ is Bob's register. If $r$ is odd, then $B_r \equiv B_{r-1}$ and if $r$ is even, then $A_r\equiv A_{r-1}$. On input $x, y$, let the joint state in registers $A_rC_rB_r$ be $\Theta_{r,A_rC_rB_r}^{x, y}$. Then the global state at round $r$ is
\begin{eqnarray}
\label{eq:roundrstate}
\ket{\Psi_r}_{XYR_XR_YA_rC_rB_r} \defeq \sum_{x,y}\sqrt{\mu(x,y)}\ket{xxyy}_{XR_XYR_Y}\ket{\Theta_r^{x, y}}_{A_rC_rB_r}.
\end{eqnarray}

We define the following quantities.
\begin{align*}
\textrm{Worst-case error:} & \quad \err(\Pi) \defeq \max_{(x,y)} \{ \Pr[O(x,y) \neq F(x,y)] \} .\\
\textrm{Quantum CC of a protocol:}&\quad  \QCC(\Pi) \defeq \sum_i |C_i| . \\
\textrm{Quantum CC of $F$:}&\quad  \Q_\eps(F) \defeq \min_{\Pi: \err(\Pi) \leq \eps} \QCC(\Pi).  \\
\end{align*}

Our first fact  links $\err(\Pi)$ with the distance $\Delta$ between a pair of final states corresponding to inputs with different outputs.

\begin{fact}[Error vs. distance] \label{fact:deltavserror}
Consider a non-constant function $f$, and let $x, y$ and $y^\prime$ 
be  inputs such that 
$f(x, y) \not= f(x, y^\prime)$.  
For any protocol $\Pi$ with $t$ rounds, it holds that
\begin{align*}
\Delta (\Theta_{t, A_t C_t}^{x, y}, \Theta_{t, A_t C_t}^{x, y^\prime}) \geq 1 - 2\err(\Pi).
\end{align*}
\end{fact}

In below, let $A'_r, B'_r$ represent Alice and Bob's registers after reception of the message $C_r$ at round $r$. That is, at even round $r$, $A'_r = A_rC_r, B'_r=B_r$ and at odd $r$, $A'_r=A_r, B'_r=B_rC_r$. We will need the following version of the quantum-cut-and-paste lemma from \cite{NT16} (also see \cite{JRS03, JN14} for similar arguments). This is a special case of \cite[Lemma 7]{NT16} and we have rephrased it using our notation. 

\begin{lemma} [Quantum cut-and-paste] \label{lem:quantum_cut_paste} Let $\Pi$ be a quantum protocol with classical inputs and consider distinct inputs $u,u'$ for Alice and $v,v'$ for Bob. Let $\ket{\Psi_{0,A_0 B_0}}$ be the initial shared state between Alice and Bob. Also let $\ket{\Psi_{k,A'_k B'_k}^{u'',v''}}$ be the shared state after round $k$ of the protocol when the inputs to Alice and Bob are $(u'',v'')$ respectively. For $k$ odd, let
$$
h_k = \BR \left( \Psi_{k,B'_k}^{u,v}, \Psi_{k,B'_k}^{u',v}\right)
$$
and for even $k$, let
$$
h_k = \BR \left( \Psi_{k,A'_k}^{u,v}, \Psi_{k,A'_k}^{u,v'}\right) .
$$
Then 
$$
\BR \left( \Psi_{r,A'_r}^{u',v},  \Psi_{r,A'_r}^{u',v'}\right) \le 2 \sum_{k=1}^{r} h_k .
$$
\end{lemma} 

As discussed in the introduction, approximate rank lower bounds bounded-error quantum communication complexity with shared entanglement~\cite{LS08c}:

\begin{fact}
\label{fact:arankQ}
For any two-party function $F:\X \times \Y \to \B$ and $\eps\in[0,1/3]$,  we have $\Q_\eps(F) = \Omega(\log \rank_\eps(F)) - O(\log\log(|\X|\cdot |\Y|))$.
\end{fact}

\subsection{Quantum information complexity}
\label{sec:qic}

\begin{definition}
Given a quantum protocol $\Pi$ with classical inputs distributed as~$\mu$,
the \emph{quantum information cost} is defined as
\begin{align}
\QIC (\Pi, \mu) \quad = \quad &
\sum_{i~\mathrm{odd}} \I (R_X R_Y  \!:\!
C_i \,|\,  Y  B_i) \quad + \quad 
\sum_{i~\mathrm{even}} \I (R_X R_Y  \!:\!
C_i \,|\,  X A_i) \enspace.
\end{align}
\end{definition}

\begin{definition}
Given a quantum protocol $\Pi$ with classical inputs distributed as~$\mu$,
the cumulative \emph{Holevo information cost}  
is defined as
\begin{align*}
\HQIC (\Pi, \mu) \quad = \quad & \sum_{i~\mathrm{odd}} \I
(X \!:\! B_i C_i \,|\, 
Y ) \quad + \quad \sum_{i~\mathrm{even}} \I
(Y \!:\! A_i C_i \,|\, 
X ) \enspace.
\end{align*}
\end{definition}

\begin{definition}
Given a quantum protocol $\Pi$ and a product distribution $\mu$ over the classical inputs, the cumulative
\emph{superposed-Holevo information cost} is defined as
\begin{align*}
\SQIC (\Pi, \mu) \quad := \quad &  \sum_{i~\mathrm{odd}} \I (X \!:\! Y R_Y  B_i C_i )_{\rho_i} \quad + \quad
\sum_{i~\mathrm{even}} \I (Y \!:\! X R_X  A_i C_i )_{\rho_i} \enspace.
\end{align*}
\end{definition}

Note that for product input distributions on $XY$ and for each $i$, 
\begin{align}
 \I (X \!:\!   B_i C_i | Y )_{\rho_i} =  \I (X \!:\! Y   B_i C_i )_{\rho_i} \leq \I (X \!:\! Y R_Y  B_i C_i )_{\rho_i}, \\
\I (Y \!:\!   A_i C_i | X )_{\rho_i} =  \I (Y \!:\! X   A_i C_i )_{\rho_i} \leq \I (Y \!:\! X R_X  A_i C_i )_{\rho_i}.
\end{align}
Combining with other results in Ref.~\cite{LT17}, we get the following for any 
$t$ round protocol $\Pi$ and any product distribution $\mu$:
\begin{align}
2\QCC(\Pi)  & \geq \QIC (\Pi, \mu) \\
	& \geq \frac{1}{t} \SQIC (\Pi, \mu) \label{eq:SHQIC} \\
	& \geq \frac{1}{t} \HQIC (\Pi, \mu) \\
	& \geq \frac{1}{2t} \QIC(\Pi, \mu).
\end{align}
.

\section{Lower bound on the information complexity of $\SINK \circ \XOR$}

\subsection{Reducing Equality to $\SINK \circ \XOR$}
\label{subsec:eqtosink}
We define the Equality function as 
$$\Eq{}(x,y) = \begin{cases} 1 &\textit{ if } x=y,\\ 0 &\textit{ otherwise.}\end{cases}$$
Recall the $\SINK$ function from Definition \ref{sink}. Following~\cite{ANS18} we use projections of the inputs in our proof to analyze the input of the $\SINK$ function. Let $w \in \{0, 1\}^{m \choose 2}$. Let $E_{v_i}$ be the set of $m-1$ input coordinates that correspond to the edges incident to $v_i$. We use the notation $w_{v_i}$ to denote the input projected to the
coordinates in $E_{v_i}$. Note that $w_{v_i}$ decides whether or not $v_i$ is a sink. By $z_{v_i}$, we refer to the $m- 1$ bit
string such that $v_i$ is a sink if and only if $w_{v_i} = z_{v_i}$. $\SINK$ can be written as
$$\SINK(w) =  \vee^{m}_{i=1} \Eq{}(w_{v_i} ,z_{v_i})$$
since only one of the vertex can be a sink in the complete directed graph. Our communication function is $\SINK \circ \XOR: \{ 0,1 \}^{m \choose 2} \times \{ 0,1 \}^{m \choose 2} \to \{0,1\}$. Similar to $\SINK$, $\SINK\circ\XOR$ can be represented as 
$$\SINK\circ\XOR(x,y) =  \vee^{m}_{i=1} \Eq{}(x_{v_i} , y_{v_i} \oplus z_{v_i}).$$ 
Our first result is as follows.
\begin{theorem}
\label{theo:SINKtoEQcl}
Suppose $m\geq 10$. Let $\Pi$ be a protocol for $\SINK\circ\XOR$ which makes a worst case error of at most $\frac{1}{4}$. There exists a protocol $\Pi'$ for $\Eq{}$ that makes a worst case error of at most $\frac{1}{4}+ \frac{m-1}{2^{m-2}} \leq \frac{1}{3}$. Furthermore, it holds that 
$$\IC(\Pi', \nu) \leq \frac{2}{m}\IC(\Pi, \mu),$$ where $\nu$ is the uniform distribution over inputs to $\Eq{}$ and $\mu$ is uniform over the inputs to $\SINK\circ\XOR$.
\end{theorem}
\begin{proof}
We have
 $$\IC(\Pi, \mu) =\I(X:\Pi |YR\srd_B) +\I(Y:\Pi |XR\srd_A) = \I(X:\Pi YR\srd_B) + \I(Y:\Pi XR \srd_A),$$
where the information quantities are evaluated on $\mu$ and the associated $\Pi$. Let $S$ be a random variable which takes values in $\{E_{v_1}, E_{v_2},\ldots , E_{v_m}\}$ with uniform probability. Let $X_{E_{v_i}}$ (similarly $Y_{E_{v_i}}$) be the restriction of $X$ (similarly $Y$) to coordinates in $E_{v_i}$. Since each coordinate $j$ appears in exactly two sets in $\{E_{v_1}, E_{v_2},\ldots , E_{v_m}\}$, we have $\Pr[j\in S] = \frac{2}{m}$. Thus, from Lemma \ref{lem:shearer}, we have
\begin{eqnarray}
\frac{2}{m}\IC(\Pi, \mu)  & \geq \expec_{s}[\I(X_S: Y\Pi R \srd_B| S=s) + \I(Y_S: X\Pi R \srd_A| S=s)] \label{eq:clshearer} \\
		& = \expec_{s}[\I(X_S: \Pi|Y R \srd_B , S=s) + \I(Y_S: \Pi| X R \srd_A, S=s)].
\end{eqnarray}
The protocol $\Pi'$ for $\Eq{}$ is now as follows, for inputs $c,d \in \{0,1\}^{m-1}$ (we use $c,d$ as inputs here to avoid confusion with $x,y$ for $\SINK\circ\XOR$). 
\begin{itemize}
\item Alice and Bob take a sample $s$ from $S$ using shared randomness. Let $i$ be such that $E_{v_i}=s$.
\item They set $x_s=c$ and $y_s=d\oplus z_{v_i}$. Alice  samples $x_{\bar{s}}$ uniformly at random from  private randomness and Bob samples $y_{\bar{s}}$ uniformly at random from private randomness. Here $\bar{s}$ is the complement of $s$. This specifies the input $x,y$ for $\SINK\circ\XOR$. 
\item They run the protocol $\Pi$ and output accordingly. 
\end{itemize}
Observe that $x_s$ and $y_s$ are distributed uniformly if $c$ and $d$ are. Thus, 
\begin{eqnarray*}
\IC(\Pi', \nu) & =  \expec_{s}[\I(X_S: \Pi|Y_S R Y_{\bar{S}} \srd_B, S=s) + \I(Y_S: \Pi| X_S R X_{\bar{S}} \srd_A, S=s)] \\
		& = \expec_{s}[\I(X_S: \Pi|Y R \srd_B, S=s) + \I(Y_S: \Pi| X R \srd_A, S=s)] ,
\end{eqnarray*}
where the information quantities are evaluated on $\mu$ and the associated $\Pi$, 
and the desired information bound follows by (\ref{eq:clshearer}).

To bound the worst case error of $\Pi'$, we argue as follows.
Fix some input $c, d$ to $\Pi^\prime$.
 If $c=d$, then $x_s = y_s \oplus z_{v_i}$ which implies that error of $\Pi'$ on this input is same as the error of $\Pi$ on the corresponding $x,y$, hence at most $\err (\Pi)$.
Now consider the case where $c\neq d$. The function $\SINK\circ\XOR$ evaluates to 1 only if $x_{E_{v_j}} =y_{E_{v_j}} \oplus z_{v_j}$ for some $j \in [m]$. Since, $c\neq d$, we conclude that $j$ (if it exists) cannot be equal to $i$. 
Moreover, the edge adjacent to $i$ is already fixed by $c, d$, and if it is not consistent with the corresponding value in $z_{v_i}$, then $j$ is not a sink. Hence, similar to 
 the argument in \cite[Claim 5.6]{ANS18}, 
the probability that $j$ is a sink is at most $\frac{1}{2^{m-2}}$, as all $m-1$ edges must be incoming 
and the edge adjacent to $i$ is already fixed.  Hence by a union bound, the probability 
for an $x,y$ (that satisfy $x_{v_i}=c ,y_{v_i}=d\oplus z_{v_i}, c \not= d$) to 
form a $1$ input at some other coordinate $j$ is at most $\frac{m-1}{2^{m-2}}$. 
This implies that $\err(\Pi') \leq \err(\Pi) + \frac{m-1}{2^{m-2}}$. This completes the proof.
\end{proof}

\subsection{Lower bound on information complexity of Equality}
To complete the argument, we use the following lemma (that uses a cut and paste argument) implicit in \cite{ABB+16a} and obtain a lower bound on the information complexity of $\Eq{}$. We repeat its proof for completeness (and consistency with our notation).
\begin{lemma}
\label{lem:cutpasteEQ}
Let $\Pi$ be a protocol for $\Eq{}$ that makes a worst case error of at most $\frac{1}{3}$. Then it holds that $\IC(\Pi, \nu)\geq \frac{1}{432}$, where $\nu$ is uniform over inputs to $\Eq{}$.
\end{lemma}
\begin{proof}
Let $\srd_A$ and $\srd_B$ be private randomness of Alice and Bob (respectively) in the protocol and $R$ be the public randomness. We have $$\IC(\Pi, \nu) = \I(Y: \Pi| X\srd_AR)   +    \I(X: \Pi| Y\srd_BR).$$ 
By the average-encoding theorem (Fact \ref{lem:avenc}), it holds that
\begin{eqnarray*}
\I(X: \Pi~|~ Y\srd_BR)   &=& \I(X: R \srd_B \Pi ~|~ Y) \\
				&  \geq & \I(X:  \Pi ~|~ Y)  \\ 
     &\geq&  \expec_{x,y \leftarrow XY}\BR^2((\Pi)^{x,y},\Pi^y). 
\end{eqnarray*}
Similarly,
\begin{eqnarray*}
\I(Y: \Pi~|~ X\srd_AR)   &=& \I(Y: X\srd_AR \Pi) \\
				& \geq & \I(Y:  \Pi)  \\ 
     &\geq&  \expec_{y \leftarrow Y}\BR^2((\Pi)^{y},\Pi). 
\end{eqnarray*}
Using the weak triangle inequality (Fact \ref{fact:triangle}), the above two inequalities imply
\begin{eqnarray*}
\expec_{x,y \leftarrow XY}\BR^2((\Pi)^{x,y},\Pi) &\leq& 2\expec_{x,y \leftarrow XY}(\BR^2((\Pi)^{x,y},\Pi^y) + \BR^2((\Pi)^{y},\Pi))\\
&\leq& 2(\I(X: \Pi~|~ Y\srd_BR) + \I(Y: \Pi~|~ X\srd_AR)) \\
& = & 2\IC(\Pi, \nu).
\end{eqnarray*}
Since $x,y$ are uniform, we can write the above relation as
$$\expec_{t\leftarrow Y}\expec_{x \leftarrow X}\BR^2((\Pi)^{x,x\oplus t},\Pi) \leq 2\IC(\Pi, \nu).$$ Since $\Pr[t=0] = \frac{1}{2^{m-1}}$, this implies that there exists an $t\neq 0$ such that
$$\expec_{x \leftarrow X}\BR^2((\Pi)^{x,x\oplus t},\Pi) \leq 3\IC(\Pi, \nu).$$
An equivalent way to write the above inequality, by relabeling $x\rightarrow x\oplus t$, is
$$\expec_{x \leftarrow X}\BR^2((\Pi)^{x\oplus t,x},\Pi) \leq 3\IC(\Pi, \nu).$$
By the weak triangle inequality (Fact \ref{fact:triangle}), we conclude
$$\expec_{x \leftarrow X}\BR^2((\Pi)^{x\oplus t,x},\Pi^{x, x\oplus t}) \leq 12\IC(\Pi, \nu).$$
The pythagorean property (Fact \ref{lemma:classicalpythagorean}) now implies that
$$\expec_{x \leftarrow X}\BR^2((\Pi)^{x,x},\Pi^{x, x\oplus t}) \leq 24\IC(\Pi, \nu).$$
Thus, there exists some $x$ for which $\BR^2((\Pi)^{x,x},\Pi^{x, x\oplus t}) \leq 24\IC(\Pi, \nu)$. Since $\Pi$ makes an error of at most $\frac{1}{3}$, we require (using relation between Bures metric and triangle inequality, Fact \ref{fact:deltabures})
$$\BR^2((\Pi)^{x,x},\Pi^{x, x\oplus t}) \geq \frac{1}{2}\Delta^2((\Pi)^{x,x},\Pi^{x, x\oplus t}) \geq \frac{1}{18}.$$ Thus, $\IC(\Pi, \nu) \geq \frac{1}{432}$, which completes the proof. 
\end{proof}

Theorem \ref{theo:SINKtoEQcl} and Lemma \ref{lem:cutpasteEQ} jointly imply that $\IC(\Pi, \mu) \geq \frac{m}{864}$, for any protocol $\Pi$ that makes an error of at most $\frac{1}{4}$ on $\SINK\circ \XOR$. This establishes the desired lower bound.  


\section{Reducing Equality to $\SINK$ for quantum information}

\subsection{Shearer-type embedding}
\label{sec:qu-shearer}







We begin by showing a general embedding result based on the Shearer-type lemma for quantum information (Lemma \ref{lem:shearer}). Consider a protocol $\Pi$ acting on input registers $X_1, X_2, \ldots, X_m$ and $ Y_1, Y_2, \ldots Y_m $, with $X_1 \equiv X_2 \equiv \ldots \equiv X_m$ and $ Y_1 \equiv Y_2 \equiv \ldots \equiv Y_m$. Define $X = X_1 X_2 \ldots X_m$, $Y = Y_1 Y_2 \ldots Y_m$. Consider a product input distribution $\mu = \mu_1 \otimes \mu_2$  on $X_i Y_i$. Consider $t \in [m]$ and let $S = \{i_1, i_2, \ldots, i_t \} \subseteq [m]$ be a random set of size $t$ picked independently of the input on $XY$ and satisfying $\prob{i\in S}\leq\frac{1}{k}$ for all $i$.  Let $X_S = X_{i_1} X_{i_2} \ldots X_{i_t}$, $Y_S = Y_{i_1} Y_{i_2} \ldots Y_{i_t}$. We define the following protocol $\Pi_S$ acting on input $A_{in} B_{in}$, with $A_{in} \equiv X_S$, $B_{in} \equiv Y_S$.

\begin{framed}
\textbf{Protocol $\Pi_S$ on input $\sigma_{A_{in} B_{in}}$}
\begin{enumerate}
\item Alice privately sample $X_i$ for each $i \not\in S$ as $\ket{\mu_1}_{X_i R_{X_i}}$.
\item Bob privately sample $Y_i$ for each $i \not\in S$ as $\ket{\mu_2}_{Y_i R_{Y_i}}$.
\item Alice embeds $A_{in}$ into $X_S$.
\item Bob embeds $B_{in}$ into $Y_S$.
\item They run $\Pi$, and output $\Pi$'s output.
\end{enumerate}

\end{framed}

\begin{lemma}
\label{lem:qu-shearer-emb1}

\begin{align*}
\Pi_S (\sigma_{A_{in} B_{in}}) & =  \Pi (\sigma_{X_S Y_{S}} \otimes (\rho_\mu^{\otimes m-t})_{X_{\bar{S}} Y_{\bar{S}}}), \\
\SQIC (\Pi_S, \mu^{\otimes t}) &  = \sum_{i~odd} \I (X_S \!:\! Y R_Y  B_i C_i )_{\rho_i} + \sum_{i~even} \I (Y_S \!:\! X R_X  A_i C_i )_{\rho_i},
\end{align*}
with $\rho_i$ the state in round $i$ when $\Pi$ is run on input distribution $\mu^{\otimes m}$.
\end{lemma}

\begin{proof}

By the definition of protocol $\Pi_S$, the channel it implements is $\Pi (\sigma_{X_S Y_{S}} \otimes (\rho_\mu^{\otimes m-t})_{X_{\bar{S}} Y_{\bar{S}}})$ (see  (\ref{eq:classicalst}) in Section~\ref{sec:comm} for definition of $\rho_\mu$) on input $\sigma_{A_{in} B_{in}}$.

For the information cost when $\Pi_S$ is run on input distribution $\mu^{\otimes t}$, first notice that for a given $S$, we can rewrite $Y R_Y = Y_S R_{Y_S} Y_{\bar{S}} R_{Y_{\bar{S}}}$. After embedding $A_{in} B_{in}$ into $X_S Y_S$, the $X_S Y_S$ registers correspond to the input of $\Pi_S$ while $R_{X_S} R_{Y_S}$ correspond to the purification of the input registers. The $X_{\bar{S}} R_{X_{\bar{S}}}$ and $Y_{\bar{S}} R_{Y_{\bar{S}}}$ registers correspond to the part privately sampled according to $\mu = \mu_1 \otimes \mu_2$ by Alice and Bob, respectively, in order to run $\Pi$. Hence, for a given $S$, the terms in $\SQIC$ look like
\begin{align*}
\I(X_S : Y_S R_{Y_S}   Y_{\bar{S}} R_{Y_{\bar{S}}} B_i C_i) = \I(X_S : Y R_Y B_i C_i), \\
\I(Y_S : X_S R_{X_S}   X_{\bar{S}} R_{X_{\bar{S}}} A_i C_i) = \I(Y_S : X R_X A_i C_i).
\end{align*}
The result follows.

\end{proof}

Let $\ket{\phi_S}_{S_A S_B}$ be a quantum state shared between Alice and Bob and encoding the distribution on $S$. 
Given $S$, let $P_A^S$ and $P_B^S$ be permutations (over the computational basis) acting on $A_{in}$ and $B_{in}$, respectively, and such that $\mu$ is invariant under their action, i.e.
\begin{align}
\label{eq:invariance}
(P_A^S \otimes P_B^S) (\rho_\mu^{\otimes t}) = \rho_\mu^{\otimes t}.
\end{align}
We define the following protocol $\hat{\Pi}$ also acting on $A_{in} B_{in}$.

\begin{framed}
\textbf{Protocol $\hat{\Pi}$ on input $\sigma_{A_{in} B_{in}}$}
\begin{enumerate}
\item Alice and Bob share $\ket{\phi_S}_{S_A S_B}$.
\item Conditioned on the value of $S$ shared in $\ket{\phi_S}$, Alice and Bob apply $P_A^S$ and $P_B^S$ to their inputs, respectively.
\item Conditioned on value of $S$ shared in $\ket{\phi_S}$, Alice and Bob run $\Pi_S$, and output $\Pi_S$'s output.
\end{enumerate}

\end{framed}

\begin{lemma}
\label{lem:qshearer}
\begin{align*}
\hat{\Pi} (\sigma_{A_{in} B_{in}}) & = \mathbb{E}_S [ \Pi_S \circ (P_A^S \otimes P_B^S) (\sigma_{A_{in} B_{in}}) ], \\
\SQIC (\hat{\Pi}, \mu^{\otimes t}) &  = \mathbb{E}_S \SQIC (\Pi_S, \mu^{\otimes t}) \leq \SQIC (\Pi, \mu^{\otimes m}) / k.
\end{align*}
\end{lemma}

\begin{proof}

By the definition of protocol $\hat{\Pi}$, the channel it implements is $\mathbb{E}_S [ \Pi_S \circ (P_A^S \otimes P_B^S) ]$.

For the information cost, let $\hat{\rho}_i $ be the state in round $i$ when $\hat{\Pi}$ is run on input distribution $\mu^{\otimes t}$.
Similar comments in the proof of Lemma~\ref{lem:qu-shearer-emb1} hold regarding $XY$ vs. $X_S Y_S X_{\bar{S}} Y_{\bar{S}}$ and the corresponding $R$ purification registers.
Hence the terms for $\SQIC$ look like
\begin{align}
\I(X_S : S_B Y R_Y B_i C_i)_{\hat{\rho}_i} & = \I(X_S :  Y R_Y B_i C_i | S)_{\hat{\rho}_i} \\
 		&= \mathbb{E}_S \I(X_S :  Y R_Y B_i C_i )_{\hat{\rho}_{i}^S},
\end{align}
where $\hat{\rho}_i^S$ is the state on registers other than $S_AS_B$, conditioned on $S$. 
Let $P_{A, X_S}^S, P_{A, R_{X_S}}^S$ (similarly $P_{B, Y_S}^S, P_{B, R_{Y_S}}^S$) be the operator $P_A^S$ (similarly $P_B^S$) acting on the registers $X_S, R_{X_S}$ (similarly $Y_S, R_{Y_S}$) respectively. Then, for any $S$,  Equation \ref{eq:invariance} implies that
\begin{align}
(P_{A, X_S}^S \otimes P_{B, Y_S}^S) (P_{A, R_{X_S}}^S \otimes P_{B, R_{Y_S}}^S) \ket{\mu^{\otimes t}}_{X_S R_{X_S} Y_S R_{Y_S}} = \ket{\mu^{\otimes t}}_{X_S R_{X_S} Y_S R_{Y_S}}.
\end{align}
Recall that $\rho_i$ is the state in round $i$ when $\Pi$ is run on input distribution $\mu^{\otimes m}$. Thus
\begin{align}
(P_{A, R_{X_S}}^S \otimes P_{B, R_{Y_S}}^S) (\hat{\rho}_{i}^S) = \rho_i
\end{align}
is independent of $S$, since the operations on the $R$ registers commute with the operations in protocol $\Pi$. By invariance of mutual information under local unitaries, we get
\begin{align}
\mathbb{E}_S \I(X_S :  Y R_Y B_i C_i )_{\hat{\rho}_{i}^S} & = \mathbb{E}_S \I(X_S :  Y R_Y B_i C_i )_{\rho_{i}} \\
		& = \I(X_S :  Y R_Y B_i C_i |S )_{\rho_{i}},
\end{align}
in which we also used that $S$ is picked independently of the input and thus stays independent of $\rho_i$ throughout.
Similar results hold for the terms accounting for Alice's information about Bob's input in $\SQIC$.
It follows that $\SQIC (\hat{\Pi}, \mu^{\otimes t}) = \mathbb{E}_S \SQIC(\Pi_S, \mu^{\otimes t})$.

To relate this to $\SQIC (\Pi, \mu^{\otimes m})$, we apply the Shearer type lemma for quantum information (Lemma~\ref{lem:shearer}) to get
\begin{align*}
\I(X_S : Y R_Y B_i C_i | S)_{\rho_i} \leq \frac{1}{k} \I(X: Y R_Y B_i C_i)_{\rho_i}, \\
\I(Y_S : X R_X A_i C_i | S)_{\rho_i} \leq \frac{1}{k} \I(Y: X R_X A_i C_i)_{\rho_i},
\end{align*}
and the result follows.

\end{proof}

\subsection{From $\SINK \circ \XOR$ to $\Eq{}$}

We get the following theorem relating $\SQIC$ for $\SINK \circ \XOR$ and $\Eq{}$.

\begin{theorem}
\label{th:sink-embed}
Fix a $t$ round quantum communication protocol $\Pi$ making worst-case error $\epsilon$ on function $\SINK \circ \XOR$ for inputs of size ${m \choose 2}$ bits. Then there exists a $t$ round quantum communication protocol $\Pi_E$ making worst case error $\epsilon + o(1)$ on $\Eq{}$ with inputs of size $m-1$ bits and satisfying the following for $\nu$ the uniform distribution on $1 + 1$ bits :
\begin{align*}
\SQIC (\Pi_E, \nu^{\otimes m-1}) \leq \frac{2}{m} \SQIC(\Pi, \nu^{\otimes {m \choose 2}}).
\end{align*}
\end{theorem}

\begin{proof}
Recall the sets $E_{v_i}$, for $i \in [m]$, as defined in Subsection \ref{subsec:eqtosink}. In the setting of the Shearer-type embedding above (Lemma \ref{lem:qshearer}), pick $S = E_{v_i}$ with probability $1/m$ for each $i \in [m]$. Let $P_A^{S_i}$ be the map that performs bit-wise addition $\oplus z_{v_i}$, and $P_B^{S_i}$ is the identity. Notice that each pair $(k, l)$, for $k < l$, appears for exactly two choices of $i$: once for $i = k$, and once for $i = l$. Hence, $\Pr [l \in S] \leq 2/m$ for all $l \in [m]$, and $2/m$ is the probability we use in the Shearer-type embedding.
By using  $\nu$ the uniform distribution on $1+1$ bits as the product distribution $\mu$ in the Shearer-type embedding, the $\SQIC$ bound follows.

It is left to argue that the resulting protocol $\Pi_E$ taken to be $\hat{\Pi}$ of the embedding is good at solving EQ. But this follows as in the classical embedding argument (see the proof of Theorem \ref{theo:SINKtoEQcl}) since the probability that Alice and Bob privately sampled inputs to $\Pi$ on $\bar{S}$ that already make $\SINK \circ \XOR$ evaluate to $1$ on $\bar{S}$  is exponentially small in $m$, hence the additional error is $o(1)$.
\end{proof}


\subsection{Quantum information cost of Equality function}

We use the following lemma about the quantum information cost of the equality function $\Eq{}$ on the uniform distribution, which was implicitly shown via a quantum cut and paste argument in Ref.~\cite{ABGJKL17}. 
\begin{lemma}
\label{theo:eqinflb}
Fix a $t$ round quantum communication protocol $\Pi$ making  worst-case error at most $\frac{1}{3}$ on $\Eq{}$. Let $\ket{\Psi_r}_{XYR_XR_YA_rC_rB_r}$ be the quantum state in $r$-th round, as defined in (\ref{eq:roundrstate}) in Section~\ref{sec:comm}, when $\Pi$ is run on the uniform distribution $\mu^{\otimes k}$ on $k + k$ bits. It holds that
$$\HQIC (\Pi, \mu^{\otimes k}) \geq \frac{1}{40000 t}.$$
\end{lemma}

The proof of our main result, Theorem \ref{theo:main}, follows.

\begin{proof}[Proof of Theorem \ref{theo:main}]
 Let $\Pi$ be a $t$-round protocol for  $\SINK\circ\XOR$  making worst-case error at most $1/5$ on input graphs of size $m$, for $m$ large enough. Then by Theorem~\ref{th:sink-embed} there exists a $t$-round protocol $\Pi_E$ for $\Eq{}$  making error at most $1/3$ and with information cost satisfying 
\begin{align*}
\SQIC (\Pi, \mu^{\otimes {m \choose 2}}) \geq \frac{m}{2} \SQIC (\Pi_E, \mu^{\otimes m-1}),
\end{align*}
with $\mu$ the uniform distribution on $1 + 1$ bits. Combining with Lemma~\ref{theo:eqinflb} and (\ref{eq:SHQIC}), the following chain of inequality gives the result:
\begin{align*}
 \frac{2 t}{m} \QIC(\Pi, \mu^{\otimes {m \choose 2}}) & \geq \frac{2}{m} \SQIC(\Pi, \mu^{\otimes {m \choose 2}})\\
		&  \geq \SQIC (\Pi_E, \mu^{\otimes m-1}) \\
		& \geq \HQIC (\Pi_E, \mu^{\otimes m-1})  \\
		& \geq \frac{1}{40000 t}.
\end{align*}
\end{proof}

We add the proof of Lemma \ref{theo:eqinflb} for completeness.

\begin{proof}[Proof of Lemma \ref{theo:eqinflb}]
By averaging over the conditioning register and then applying the average encoding theorem (Fact~\ref{fact:IvsB} and Lemma~\ref{lem:avenc}), we conclude that
\beqar
\label{eqinflb1}
&&\HQIC (\Pi, \mu^{\otimes k}) \defeq \sum_{r = odd}\I(X:B_rC_r|Y)_{\Psi_r} + \sum_{r=even}\I(Y: A_rC_r|X)_{\Psi_r}\nonumber\\
&&\geq \expec_{x,y \leftarrow \mu}\(\sum_{r = odd}\BR\(\Psi_{r, B_rC_r}^{x,y}, \Psi_{r, B_rC_r}^y\)^2 + \sum_{r=even}\BR\(\Psi_{r, A_rC_r}^{x,y}, \Psi_{r, A_rC_r}^x\)^2\) \nonumber\\
&&\geq \frac{1}{t}\(\expec_{x,y \leftarrow \mu}\(\sum_{r = odd}\BR\(\Psi_{r, B_rC_r}^{x,y}, \Psi_{r, B_rC_r}^y\) + \sum_{r=even}\BR\(\Psi_{r, A_rC_r}^{x,y}, \Psi_{r, A_rC_r}^x\)\)\)^2 .
\eeqar
Let $x_1, x_2, y_2$ be drawn uniformly from $\B^k$ and let $y_1\defeq x_1$. Observe that, taken separately, $(x_1, y_2)$, $(x_2, y_1)$ and $(x_2, y_2)$ are distributed uniformly. Thus, (\ref{eqinflb1}) ensures that 
\beqarst
\sqrt{t\HQIC (\Pi, \mu^{\otimes k})} &\geq& \expec_{x_1,y_2 \leftarrow \mu}\(\sum_{r = odd}\BR\(\Psi_{r, B_rC_r}^{x_1,y_2}, \Psi_{r, B_rC_r}^{y_2}\) + \sum_{r=even}\BR\(\Psi_{r, A_rC_r}^{x_1,y_2}, \Psi_{r, A_rC_r}^{x_1}\)\) , \\
\sqrt{t\HQIC (\Pi, \mu^{\otimes k})} &\geq& \expec_{x_2,y_1 \leftarrow \mu}\(\sum_{r = odd}\BR\(\Psi_{r, B_rC_r}^{x_2,y_1}, \Psi_{r, B_rC_r}^{y_1}\) + \sum_{r=even}\BR\(\Psi_{r, A_rC_r}^{x_2,y_1}, \Psi_{r, A_rC_r}^{x_2}\)\) , \\
\sqrt{t\HQIC (\Pi, \mu^{\otimes k})} &\geq& \expec_{x_2,y_2 \leftarrow \mu}\(\sum_{r = odd}\BR\(\Psi_{r, B_rC_r}^{x_2,y_2}, \Psi_{r, B_rC_r}^{y_2}\) + \sum_{r=even}\BR\(\Psi_{r, A_rC_r}^{x_2,y_2}, \Psi_{r, A_rC_r}^{x_2}\)\) .\\
\eeqarst
Moreover, it holds that $\Pr\(\Eq{}(x_1, y_2)=1\) = \Pr\(\Eq{}(x_2, y_1)=1\) = \Pr\(\Eq{}(x_2, y_2)=1\) = \frac{1}{2^k} $. Thus, by first conditioning (separately) on $\Eq{} (x_1, y_2) = \Eq{} (x_2, y_1) = \Eq{} (x_2, y_2) = 0$ and then applying Markov's inequality, we find that there exists a choice of $x_1,x_2, y_2$ satisfying the non-equality conditions and such that 
\beqar
\label{eqinflb2}
5\sqrt{t \HQIC (\Pi, \mu^{\otimes k})} &\geq& \sum_{r = odd}\BR\(\Psi_{r, B_rC_r}^{x_1,y_2}, \Psi_{r, B_rC_r}^{y_2}\) + \sum_{r=even}\BR\(\Psi_{r, A_rC_r}^{x_1,y_2}, \Psi_{r, A_rC_r}^{x_1}\), \nonumber\\
5\sqrt{t \HQIC (\Pi, \mu^{\otimes k})} &\geq& \sum_{r = odd}\BR\(\Psi_{r, B_rC_r}^{x_2,y_1}, \Psi_{r, B_rC_r}^{y_1}\) + \sum_{r=even}\BR\(\Psi_{r, A_rC_r}^{x_2,y_1}, \Psi_{r, A_rC_r}^{x_2}\), \nonumber\\
5\sqrt{t \HQIC (\Pi, \mu^{\otimes k})} &\geq& \sum_{r = odd}\BR\(\Psi_{r, B_rC_r}^{x_2,y_2}, \Psi_{r, B_rC_r}^{y_2}\) + \sum_{r=even}\BR\(\Psi_{r, A_rC_r}^{x_2,y_2}, \Psi_{r, A_rC_r}^{x_2}\).
\eeqar
Applying the triangle inequality (Fact \ref{fact:triangle}) to (\ref{eqinflb2}), we conclude that
\beqarst
10\sqrt{t \HQIC (\Pi, \mu^{\otimes k})} &\geq& \sum_{r = odd}\BR\(\Psi_{r, B_rC_r}^{x_1,y_2}, \Psi_{r, B_rC_r}^{x_2,y_2}\) \nonumber\\
10\sqrt{t \HQIC (\Pi, \mu^{\otimes k})} &\geq& \sum_{r=even}\BR\(\Psi_{r, A_rC_r}^{x_2,y_1}, \Psi_{r, A_rC_r}^{x_2, y_2}\). \nonumber\\
\eeqarst
Assume that $t$ is even and Alice produces the output, we use the quantum cut-and-paste Lemma (Lemma \ref{lem:quantum_cut_paste}) to conclude that  
\begin{align*}
\BR\(\Psi_{t, A_tC_t}^{x_1,y_2}, \Psi_{t, A_tC_t}^{x_1, y_1}\) & \leq 2(10\sqrt{t \HQIC (\Pi, \mu^{\otimes k})} + 10\sqrt{t \HQIC (\Pi, \mu^{\otimes k})}) \\
& = 40\sqrt{t \HQIC (\Pi, \mu^{\otimes k})}.
\end{align*}
If $ \HQIC (\Pi, \mu^{\otimes k}) \leq \frac{1}{40000 t}$, we conclude that $40\sqrt{t \HQIC (\Pi, \mu^{\otimes k})} \leq \frac{1}{5}$, and then 
\begin{align*}
1 - 2\err(\Pi) & \leq \Delta(\Psi_{t, A_tC_t}^{x_1,y_2}, \Psi_{t, A_tC_t}^{x_1, y_1}) \\
		& \leq \sqrt{2} B(\Psi_{t, A_tC_t}^{x_1,y_2}, \Psi_{t, A_tC_t}^{x_1, y_1}) \\
		& \leq \sqrt{2}/5 \\
		& < 1/3,
\end{align*}
which leads to contradiction with the fact that protocol $\Pi$ makes an error of at most $\frac{1}{3}$. This completes the proof. 
\end{proof}


\section{Conclusion and open problems}
\label{sec:conclusion}

Our main result exhibits that the function introduced in \cite{ANS18} witnesses an exponential separation between quantum communication complexity and log-approximate rank. A consequence of our lower bound is that the randomized and quantum communication complexities of this function are polynomially related. Thus, the long-standing problem of finding a total function, that provides an exponential separation between randomized communication complexity and quantum communication complexity, remains open.

An interesting question that our techniques do not resolve is if we can show a round independent exponential separation between log-approximate rank and $\QIC$. We believe that it would be surprising if the log-approximate rank and $\QIC$ were polynomially related. Known functions witnessing exponential separation between $\QIC$ and $\QCC$ have a completely different structure \cite{GanorKR:2015, RaoS:2015, ATYY16}. 

Further, we would like to understand if the Shearer-type embedding can go beyond product input distributions, and if it can be improved for $\QIC$. Finally,  it would be interesting if the lower bound in Corollary \ref{cor:qlb} could be improved to $\Omega (m^{1/2})$, matching the achievable protocol using distributed Grover search (up to logarithmic terms; see \cite[Conclusion]{ANS18}).

\section*{Acknowledgements}

We thank Ashwin Nayak for detailed discussions about the proof. N.G.B thanks Shalev Ben-David for pointing out \cite{ANS18}, A.A. thanks N.G.B. for pointing out \cite{ANS18} and finally, D.T. thanks A.A. for pointing out \cite{ANS18}. We also thank Ronald de Wolf and Makrand Sinha for helpful correspondence. 

This work was done when N.G.B was visiting the Institute for Quantum Computing, University of Waterloo and supported by a Queen Elizabeth Scholarship. N.G.B is also supported by the National Research Foundation, Prime Minister’s Office, Singapore and
the Ministry of Education, Singapore under the Research Centres of Excellence programme.
A.A. and D.T. are supported in part by NSERC, CIFAR, Industry Canada. D.T. is also supported by an NSERC PDF. IQC and PI are supported in part by the Government of Canada and the Province of Ontario.


\DeclareUrlCommand{\Doi}{\urlstyle{sf}}
\renewcommand{\path}[1]{\small\Doi{#1}}
\renewcommand{\url}[1]{\href{#1}{\small\Doi{#1}}}
\newcommand{\eprint}[1]{\href{http://arxiv.org/abs/#1}{\small\Doi{#1}}}
\bibliographystyle{alphaurl}
\phantomsection\addcontentsline{toc}{section}{References} 
\bibliography{logrank}

\end{document}